\documentclass[journal,twoside,web]{ieeecolor}
\usepackage{lcsys}
\usepackage{cite}
\usepackage{comment}
\usepackage{amsmath,amssymb,amsfonts}

\usepackage{etoolbox}
\makeatletter
\@ifundefined{color@begingroup}%
  {\let\color@begingroup\relax
   \let\color@endgroup\relax}{}%
\def\fix@ieeecolor@hbox#1{%
  \hbox{\color@begingroup#1\color@endgroup}}
\patchcmd\@makecaption{\hbox}{\fix@ieeecolor@hbox}{}{\FAILED}
\patchcmd\@makecaption{\hbox}{\fix@ieeecolor@hbox}{}{\FAILED}
\usepackage{graphicx}
\usepackage{textcomp}
\def\BibTeX{{\rm B\kern-.05em{\sc i\kern-.025em b}\kern-.08em
    T\kern-.1667em\lower.7ex\hbox{E}\kern-.125emX}}
\def \bF {\pmb{F}}

\def \bw {\pmb{w}}
\def \bX {\pmb{X}}

\def \bv {\pmb{v}}

\def \bp {\pmb{p}}
\def \by {\pmb{y}}
\def \bu {\pmb{u}}
\def \bq {\pmb{q}}

\def \bb {\pmb{b}}
\def \bk {\pmb{k}}
\def \bh {\pmb{h}}

\def \bl {\pmb{l}}

\def \bzero {\pmb{0}}
\def \bone {\pmb{1}}

\def \cE {\mathcal{E}}
\def \cV {\mathcal{V}}

\def \vec {\textnormal{vec}}

\usepackage{graphicx}
\usepackage{dcolumn}
\usepackage{bm}
\usepackage[normalem]{ulem}
\usepackage{comment}
\usepackage{float}
\usepackage{mathtools}

\usepackage{hyperref}
\usepackage{multirow}
\usepackage[utf8]{inputenc}

\usepackage{amsthm}

\newtheorem{theorem}{Theorem}
\newtheorem{lemma}{Lemma}
\newtheorem{definition}{Definition}
\newtheorem{remark}{Remark}
\newtheorem{property}{Property}
\newtheorem{assumption}{Assumption}
\newtheorem{proposition}{Proposition}

\let\labelindent\relax
\usepackage[shortlabels]{enumitem}

\graphicspath{{./figures/}}

\theoremstyle{definition}

\newcolumntype{C}[1]{>{\centering\arraybackslash$}m{#1}<{$}}
\newlength{\mycolwd}                                         
\usepackage{gensymb}

\usepackage{algorithm}
\usepackage{algpseudocode}

\usepackage{varwidth}

\begin{document}

\title{Single-Integrator Consensus Dynamics over Minimally Reactive Networks}

\author{Amirhossein Nazerian,~\IEEEmembership{Graduate Student Member,~IEEE,} David Phillips, Hernán A. Makse, and Francesco Sorrentino,~\IEEEmembership{Senior Member,~IEEE}
\thanks{A. Nazerian, D. Phillips, and F. Sorrentino are with the Mechanical Engineering Department, University of New Mexico, Albuquerque, NM, 87131 USA. (email: \href{mailto:anazerian@unm.edu}{anazerian@unm.edu}; \href{mailto:dphillips1@unm.edu}{dphillips1@unm.edu}; \href{mailto:fsorrent@unm.edu}{fsorrent@unm.edu})}%
\thanks{H. A. Makse is with Levich Institute and Physics Department, City College of New York, New York, New York 10031, USA. (email: \href{mailto:hmakse@ccny.cuny.edu}{hmakse@ccny.cuny.edu})}
}

\maketitle
\thispagestyle{empty}

\begin{abstract}
    The problem of achieving consensus in a network of connected systems arises in many science and engineering applications. 
    \textcolor{black}{In contrast to previous works, we focus on the system reactivity, i.e., the initial amplification of the norm of the system states.}
    We identify a class of networks that we call minimally reactive, which are such that the indegree and the outdegree of each node of the network are the same.
    We propose several optimization procedures in which minimum perturbations (links or link weights) are imposed on a given network topology to make it minimally reactive.
    {
    A new concept of structural reactivity is introduced which measures how much
    a given network is far from becoming minimally reactive by
    link perturbations. 
    The structural reactivity of directed random graphs is studied.
    }
\end{abstract}


\section{Introduction}

\IEEEPARstart{T}{he consensus} problem is relevant to all those applications for which it is desired that the states of several dynamical systems or agents reach an agreement, such as the distributed control of multi-agent systems \cite{amirkhani2022consensus}, formation control \cite{Nuño2022Leaderless}, flocking \cite{Qin2022Multiagent}, distributed sensor networks \cite{Shen2022Consensus}, and cyber security \cite{YUAN2021Secure}.


In a general consensus problem, \textcolor{black}{the goal is for agents initialized from different initial conditions to converge to the same state, also called the consensus state.}
Previous works have established globally stable protocols to guarantee convergence, see, e.g.,  \cite{Consensus2004Olfati,olfati2007consensus,MIAO2016Collision,LIU2009Consensus}.
These papers focus on the steady-state behavior of consensus dynamics.
\textcolor{black}{A large body of work studied the finite-time stabilizing consensus problems with a prescribed performance, see, e.g., \cite{Wang2019Prescribed,NING2019Practical,Ning2023Fixed}.}
In engineering applications, it is often required that the state trajectories do not deviate from the consensus state in the transient dynamics.
\textcolor{black}{Sometimes a monotonic decrease in the norm of the state vector may be required.}

An important characterization of the transient behavior of a linear system (or linearized system about a fixed point) is provided by the concept of `reactivity' \cite{Neubert1997ALTERNATIVES}.
The reactivity measures the rate of the change of the norm of the state vector as time $t \rightarrow 0$. 
Reference \cite{Neubert1997ALTERNATIVES} showed that for an LTI system in the form $\dot{x}(t) = A x(t)$, the reactivity $R(A)$ of the matrix $A$ is equal to the largest eigenvalue of $(A + A^\top)/2$, also simply called the reactivity of the matrix $A$.
Even if the matrix $A$ is Hurwitz, its reactivity may be positive, zero, or negative.
A positive reactivity results in an initial growth of the norm of the state vector.
\textcolor{black}{
The relationship between reactivity, non-normality, and transient behavior has been the subject of study since the seminal work by Threfthen \cite{trefethen1991pseudospectra}, and more recent works such as \cite{Neubert1997ALTERNATIVES,Asllani2018Structure,Duan2022Network}.
}

In this letter, we study the reactivity of the consensus dynamics.
We prove that `minimally reactive networks' have zero reactivity.
We then propose optimizations in which the weights of a given directed and weighted network are adjusted so that the network becomes minimally reactive.
It is proven that only strongly connected networks can become minimally reactive under weight perturbation.
We also propose an optimization method to add/remove links from a given directed and unweighted network to make it minimally reactive. 

The rest of the paper is organized as follows. Notation and background information are introduced in Sec.\,\ref{sec:notation}.
The consensus problem is discussed in Sec.\,\ref{sec:consensus}.
\textcolor{black}{
The weight perturbation procedure is in Sec.\,\ref{sec:weight}.
The link perturbation procedures are presented in Sec.\,\ref{sec:link}.
}
The conclusions are provided in Sec.\,\ref{sec:conclusion}.

\section{Notation and background} \label{sec:notation}

Given a real square matrix $B$, we order its eigenvalues  such that $Re[\lambda_n(B)] \leq \hdots \leq Re[\lambda_2 (B)] \leq Re[\lambda_1(B)]$, where he notation $Re[\cdot]$ indicates the real part of its argument.
If $B$ is symmetric, $\lambda_i(B)\in \mathbb{R}, \ \forall i$.
Given two matrices $A \in \mathbb{R}^{n \times p}$ and $B \in \mathbb{R}^{m \times q}$, the expression $A \otimes B  \in \mathbb{R}^{nm \times pq}$ denotes the Kronecker product of $A$ and $B$. 
The vectorization function $\vec(\cdot): \mathbb{R}^{m \times n} \rightarrow \mathbb{R}^{mn}$ takes a matrix and returns a vector by stacking all the columns of the matrix on top of each other.

Let $G = (\mathcal{V}, \mathcal{E}, A)$ be a directed graph with set of nodes $\mathcal{V} = \{v_1, v_2, \hdots, v_n \}$, set of edges $\mathcal{E} \subseteq \mathcal{V} \times \mathcal{V}$, and a adjacency matrix $A = [A_{ij}]$ with non-negative entries. 
If there is a link going from the node $v_j$ to node $v_i$, i.e., $(v_j, v_i) \in \mathcal{E}$, then $A_{ij} > 0$, otherwise $A_{ij} = 0$.
\textcolor{black}{If $A_{ij} \in \{ 0,1 \}, \forall i,j$, then the graph $G$ is called `unweighted', otherwise, `weighted'.}
The indegree and outdegree of node $v_i$ are defined as $\deg^{in}_i = \sum_j A_{ij}$ and $\deg_i^{out} = \sum_j A_{ji}$, respectively.
A directed spanning tree of the directed graph $G$ rooted at $r$ is a subgraph $T$ of $G$ such that $T$ is a directed tree and contains a directed path from $r$ to any other vertex in $V$. 
\begin{assumption} \label{assump:G}
    We assume the directed graph $G$ has at least one directed spanning tree. 
\end{assumption}
\color{black}

The Laplacian matrix corresponding to graph $G$ is defined as $L = A - D$ where $D$ is a diagonal matrix such that $D_{ii} = \sum_j A_{ij}$.
Since the sums of the entries in the rows of the matrix $L$ are equal zero, $\lambda_1(L) = 0$ and its corresponding normalized right eigenvector $\bv_1 = 1/\sqrt{n} \bone_n$, where $\bone_n$ is a vector of ones with length $n$.
Thus, $Re[\lambda_N (L)] \leq \cdots \leq Re[\lambda_2 (L)] < \lambda_1 (L) = 0$.
\textcolor{black}{
The left eigenvector of the matrix $L$ corresponding to $\lambda_1(L) = 0$ is denoted by $\bw_1 = [w_1 \ \, w_2 \ \cdots \ w_n]^\top \in \mathbb{R}^n$.
}

The complement graph of the unweighted and directed graph $G = (\mathcal{V}, \mathcal{E})$ is defined as $H = (\mathcal{V}, \mathcal{U} - \mathcal{E})$, where $\mathcal{U} = \mathcal{V} \times \mathcal{V}$.

\section{Consensus Problem} \label{sec:consensus}
The linear single-integrator consensus problem is defined by the following set of equations that describe the time evolution of the states of the nodes of the dynamical network \cite{Consensus2004Olfati},
\color{black}
\begin{equation} \label{eq:adjecency}
    \dot{x}_i(t) = \sum_{j =1}^n A_{ij}  \left( x_j(t)-x_i(t) \right)  = \sum_j L_{ij} x_j
\end{equation}
where $x_i(t) \in \mathbb{R}$ is the state vector of the dynamical network, $A = [A_{ij}]$ and $L=[L_{ij}]$ are the adjacency matrix and the Laplacian matrix corresponding to the directed graph $G$, respectively.
The term $x_j(t)-x_i(t)$ represents linear diffusive coupling.

Based on Assumption \ref{assump:G}, system \eqref{eq:adjecency} globally converges to consensus $x_1 = x_2 =  \cdots = x_n = x_c$ for $t \rightarrow \infty$.
The consensus state $x_c = \sum_{j=1}^n w_j x_j(0)/\sum_{j=1}^n w_j$ if $\sum_{j=1}^n w_j > 0$ \cite[Corollary 2]{Consensus2004Olfati}. 
In this letter, we are mostly interested in the transient dynamics towards the consensus state, and how this is affected by
the reactivity  of the linear system \eqref{eq:adjecency} \cite{Neubert1997ALTERNATIVES}, i.e., the largest initial amplification of the norm of the vector $\bX=[x_1 \ \ x_2 \ \cdots \ x_n]^\top$,
which is equal to:
\begin{align} \label{eq:reactivity0}
\begin{split}
    R( L) & := \max_{\| \bX\| \neq 0} \left[ \left( \frac{1}{\|  \bX \|} \frac{d \|  \bX\|}{d t} \right)_{t = 0} \right] \\
    & = \max_{  \bX \neq 0} \frac{{ \bX}^\top S   \bX }{{  \bX}^\top  \bX}
     =
     \lambda_{1}\left(S \right),
\end{split}
\end{align}
where $S = ({L + L^\top})/{2}$. A positive (negative) reactivity $R( L)$ indicates that the norm of $ \bX$ tends to grow (shrink) in the limit of $t \rightarrow 0$. 
As already mentioned, the reactivity of a matrix coincides with the largest eigenvalue of the symmetric part of that matrix \cite{Neubert1997ALTERNATIVES}.
Here we are interested in minimizing the reactivity $R(L)$.
\begin{theorem}
    The reactivity of the Laplacian $L$, $R(L)$, is non-negative.
\end{theorem}
\begin{proof}
The reactivity in \eqref{eq:reactivity0} is,
\begin{equation} \label{eq:lambda}
\begin{alignedat}{3} 
    & R(L) && = \lambda_{1} (S)  \\
    &  && = \max_{ \bX \neq 0} \frac{{\bX}^\top S \bX }{{ \bX}^\top \bX} && \geq { \bv_1}^\top S \bv_1  \\
    & && && =   { \bv_1}^\top \left(\frac{L + L^\top}{2} \right) \bv_1 \\
    & && && = \frac{1}{2}(\bv_1^\top \lambda_1(L) \bv_1  + \lambda_1(L) \bv_1^\top \bv_1 ) \\
    & && && = \lambda_1(L) = 0.
\end{alignedat}
\end{equation}
Hence, $R(L) \geq 0$.  That concludes the proof.
\end{proof}
\begin{proposition}
    A network with Laplacian $L^*$ is minimally reactive, $R(L^*) = 0$, if and only if the sums over the columns of the matrix $L^*$ are zero.
\end{proposition}
\begin{proof}
    Equation \eqref{eq:lambda} suggests that the minimum possible reactivity is achieved by a network with the Laplacian $L^*$, which is such that $R(L^*) = \lambda_1 (S^*)=0$ where $S^* = (L^* + {L^*}^\top)/2$. 
    It is inferred that the inequality in \eqref{eq:lambda} is satisfied with the equal sign when $S=S^*$. 
    This means that $\bX^* = \bv_1$ becomes the maximizer in \eqref{eq:lambda}. 
    Hence, the corresponding eigenvector to $\lambda_1(S^*)$ is $\bv_1 = 1/\sqrt{n} \bone_n$. 
    We thus conclude that a necessary and sufficient condition for minimal reactivity, $R(L^*) = 0$, is that the matrix $S^*$ has zero row-sum. This in turn implies that
    \begin{equation}
        0 = \sum_j S_{ij}^* = \frac{1}{2} \left( \sum_j L_{ij}^* +  \sum_i L_{ij}^* \right) = \frac{1}{2} \left( \sum_i L_{ij}^* \right),
    \end{equation}
    i.e., that the sums over the columns of the matrix $L^*$ are zero.
    That concludes the proof.
\end{proof}
\color{black}
It follows that each node of a minimally reactive network has equal indegree and outdegree, i.e., $\deg_i^{in} = \deg_i^{out}, \ \forall i$. Networks with nodes that all have equal indegrees and outdegrees are often called `balanced' networks.
As an example, the directed network with the Laplacian
\begin{equation}
    L^* = \begin{bmatrix}
-9 & 0 & 3 & 6 \\ 
5 & -5 & 0 & 0 \\ 
0 & 5 & -5 & 0 \\ 
4 & 0 & 2 & -6 
\end{bmatrix}
\end{equation}
has minimal reactivity.
It is easy to see that any undirected network is minimally reactive.

\subsection{Properties of minimally reactive networks}
Next, we describe some of the properties of a minimally reactive network.
\begin{property}
The consensus state $x_c$ in \eqref{eq:adjecency} is equal to the average over the initial conditions, i.e., $x_c = 1/n \sum_{j=1}^n x_j(0)$. 
\end{property}
This is due to the fact that the average over the states $1/n \sum_{j=1}^n x_j$ is a constant of motion \cite[Theorem 6]{Consensus2004Olfati}.
\textcolor{black}{This also follows from the fact that the left eigenvector $\bw_1 = 1/\sqrt{n} \bone_n$.}
\begin{property}
The norm of the time trajectory $\| \bX \|$ does not increase as time increases if the network is minimally reactive.
\end{property}
This follows from \eqref{eq:reactivity0}. 
For minimally reactive networks $R(L^*) = 0$, so  $d\| \bX \| /dt \leq 0, \forall t\geq 0$.

\section{Weight perturbation} \label{sec:weight}
\color{black}
In this section, we focus on the case of weighted graphs.
Given the Laplacian matrix $L$, we aim to find a perturbation matrix $P=[P_{ij}]$ with minimum $\sum_{i,j} P^2_{ij}$ such that $L^* = [L^*_{ij}] = P + L$ is minimally reactive.
We require $L^*$ and $L$ to have the same structure, i.e.,
\begin{equation*}
    \begin{cases}
        L^*_{ij} L_{ij} = (P_{ij} + L_{ij}) L_{ij} > 0 \qquad & \textnormal{if} \ L_{ij} \neq 0, \\
        L^*_{ij} = P_{ij} + L_{ij} = 0 \qquad & \textnormal{if} \ L_{ij} = 0, 
    \end{cases} \quad \forall i,j.
\end{equation*}

\color{black}
\textcolor{black}{In order to satisfy $R(L^*) = 0$,} $L^*$ should have sums over its rows and sums over its columns equal to zero. 
Also, by definition, a weight perturbation matrix $P$ does not add links. 
Hence, the following conditions must be satisfied:
\begin{subequations} \label{eq:proof}
 \begin{gather} 
    \sum_{j=1}^n P_{ij} = 0, \quad i = 1, \hdots, n,  \label{eq:rowsum} \\
    \sum_{i=1}^n P_{ij} = - \sum_{i=1}^n L_{ij}, \quad j = 1, \hdots, n,  \label{eq:columnsum} \\
    P_{ij} = 0 \quad \textnormal{if} \quad L_{ij} = 0.  \label{eq:pzero}
\end{gather}   
\end{subequations}
Then, we define $\bp := \vec(P)$ as the vectorized matrix $P$. 
Next, we evaluate the vectorized version of \eqref{eq:proof}. The constraint over the row sum, \eqref{eq:rowsum}, is rewritten as,
$K_1 \bp = \bzero_{n^2}$ where $K_1 = \bone_n^\top \otimes I_n$, and $\bzero_n$ is a column vector of zeros of length $n$.
The constraint over the column sum, \eqref{eq:columnsum} is rewritten as,
$K_2 \bp = \bl$ where $K_2 = I_n \otimes \bone_n^\top$ and $\bl$ is the right-hand side of the above equation. 
Finally, the constraint over zero entries in $P$, \eqref{eq:pzero}, can be written as $K_3 \bp = \bzero_m$, where $m$ is the number of zero entries in $L$.
$K_3 \in \mathbb{R}^{m \times n^2}$ is a matrix with zero rows except for an entry equal to one in each row that corresponds to each $P_{ij}$ that should be zero. 
Hence, \eqref{eq:proof} is rewritten as,
\begin{equation} \label{eq:linsyseq}
    \begin{bmatrix}
        K_1 \\ K_2 \\ K_3 
    \end{bmatrix} \bp =
    \begin{bmatrix}
        \bzero_{n^2} \\ \bl \\ \bzero_m
    \end{bmatrix}, \quad \textnormal{or equivalently} \quad K\bp = \bb.
\end{equation}
\textcolor{black}{Now we deal with the condition that the nonzero entries of the optimized Laplacian $L^*$ should have the same sign as their corresponding nonzero entries of the Laplacian $L$, i.e., $L_{ij} L^*_{ij} > 0$ if $L_{ij} \neq 0$, $\forall i,j.$}
This indicates that $L^*_{ij} = L_{ij} + P_{ij} > 0$, $\forall i\neq j$.
Therefore, $P_{ij} + L_{ij} \geq \epsilon$, $\forall i\neq j, L_{ij} \not= 0$, where $\epsilon > 0$ is a small tunable parameter. 
\textcolor{black}{For $i = j$ and $L_{ii} \neq 0$, $L^*_{ii} = L_{ii} + P_{ii} < 0$, and thus, $L_{ii} + P_{ii} \leq - \epsilon$.}
These conditions can be rewritten in the vectorized form as $H \bp \leq \bh$.
To find $P$ with this new constraint, we solve the quadratic optimization in $\bp$,
\begin{align} \label{eq:optp}
\begin{split}
    {\min_{\bp}} \quad & \bp^\top \bp \\
    \textnormal{subject to} \quad & K\bp = \bb \\
    & H \bp \leq \bh.
\end{split}
\end{align}
\begin{theorem} \label{theo2}
\textcolor{black}{
The optimization problem from \eqref{eq:optp} has a solution if and only if the initial graph corresponding to the Laplacian $L$ is strongly connected.
}
\end{theorem}
\begin{proof}
\textcolor{black}{See Appendix \ref{appendix}.}
\end{proof}

\color{black}
\subsection{Example}

Here we provide an example of the application of the weight perturbation procedure to a network with positive reactivity to make it minimally reactive. 
We start with a weighted network with $n = 5$ nodes, and the Laplacian matrix
\begin{equation} \label{eq:lap}
    L = \begin{bmatrix}
        -10 & 0 & 5 & 0 & 5 \\ 
        0 & -2 & 0 & 2 & 0 \\ 
        0 & 0 & -5 & 0 & 5 \\ 
        1 & 3 & 0 & -4 & 0 \\ 
        0 & 0 & 0 & 1 & -1 
    \end{bmatrix}.
\end{equation}
The optimized Laplacian corresponding to a minimally reactive network obtained by solving the optimization in \eqref{eq:optp} is equal to
\begin{equation}
    L^* = 
    \begin{bmatrix}
        -3.95 & 0 & 3.49 & 0 & 0.46 \\ 
        0 & -1.76 & 0 & 1.76 & 0 \\ 
        0 & 0 & -3.49 & 0 & 3.49 \\ 
        3.95 & 1.76 & 0 & -5.71 & 0 \\ 
        0 & 0 & 0 & 3.95 & -3.95 
    \end{bmatrix}.
\end{equation}
We compare the time evolutions of these two networks from a randomly chosen initial condition from a standard normal distribution,
$\bX (0) = [0.0505 \ \,  0.7641 \ \,  -0.7397  \ \,   0.4984  \ \,  -1.9546]^\top$.
When the matrix $L$ is used, the consensus state is $x_c = -0.1024$, and when the matrix $L^*$ is used, the consensus states is $x_c^* = -0.2763$. 
\begin{figure}
    \centering
    \includegraphics[width=0.7\linewidth]{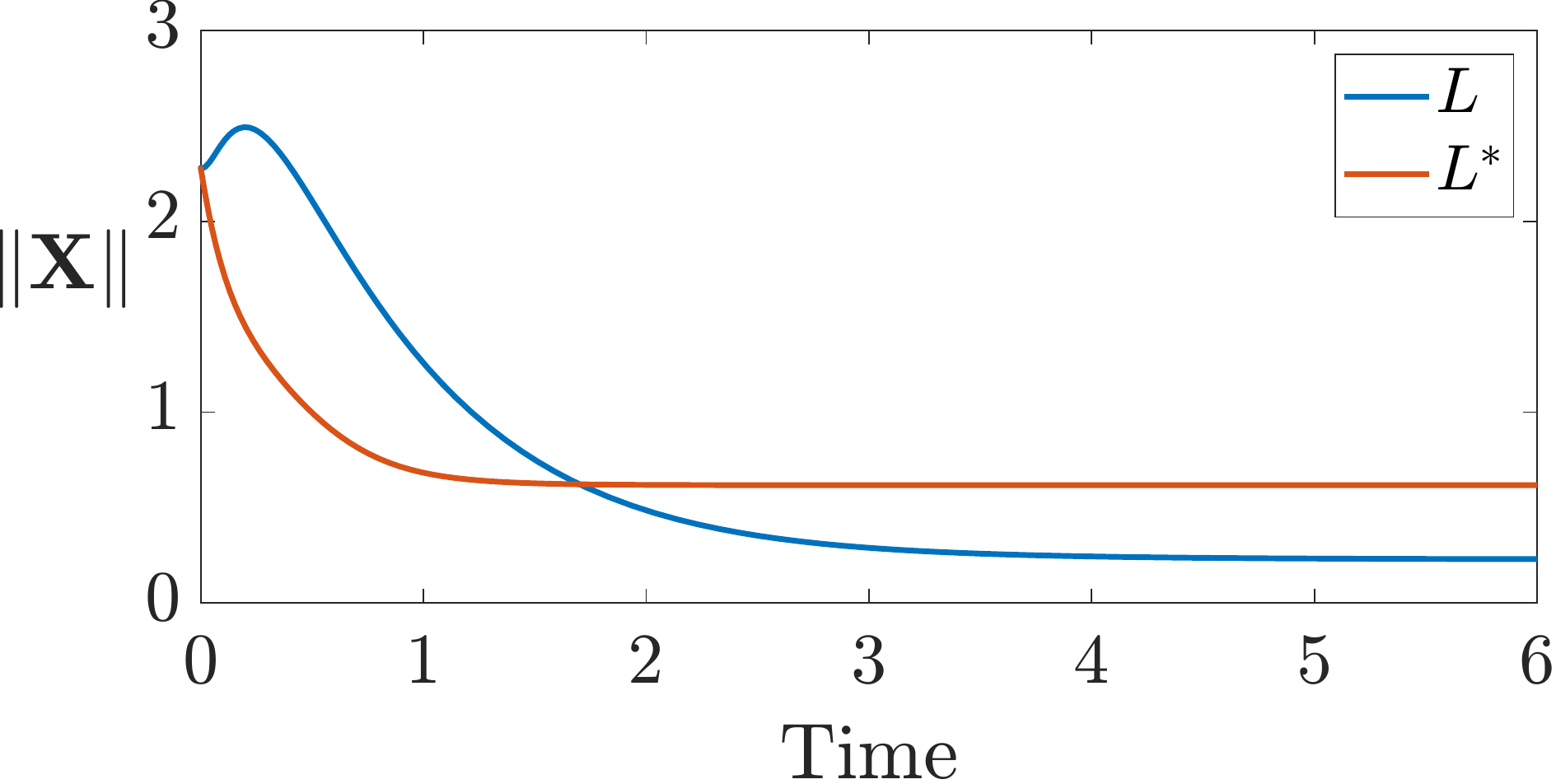}
    \caption{\textcolor{black}{
    Time evolution of the norm of the states $\bX(t)$ for two different topologies with Laplacians $L$ and $L^*$ in \eqref{eq:lap}.}
    }
    \label{fig:norm}
\end{figure}
Figure \ref{fig:norm} shows the time evolution of the norm of $\bX(t)$ for the cases of $L$ and $L^*$.
From Fig.\,\ref{fig:norm}, we see that $ \| \bX(t) \|$ increases at the initial times for the original network, while the norm monotonically approaches the norm of the consensus state for the minimally reactive network.

\color{black}


\section{Link perturbations} \label{sec:link}
In this section, we study the following problem: given a directed and unweighted graph $G$ with the Laplacian $L$, add/remove the minimum number of directed links such that the graph becomes minimally reactive, i.e., $R(L^*) = 0$, where $L^*$ is the modified Laplacian.

We discuss the solution to the above problem by only link addition, only link removal, and both link addition and removal in Sections \ref{sec:add}, \ref{sec:rem}, and \ref{sec:addrem}, respectively.
\textcolor{black}{Integer linear programs are introduced in their respective sections. 
For each case, we will show that these problems always allow for a `trivial' solution, hence they always have solutions.}

\subsection{Link addition} \label{sec:add}
Here we formulate an integer program to find the minimum number of links to be added to an unweighted and directed graph to make it minimally reactive. 
We consider an unweighted network with an adjacency matrix $A$ and introduce a minimal perturbation $P$ such that the perturbed network with the adjacency matrix $A^* = A + P$ is minimally reactive.
We assume the initial and the optimized graphs do not have self-loops.
Then, the optimal Laplacian is $L^* = A^* - D^*$, where $D^*$ is a diagonal matrix and entries on the main diagonal are the sums over the rows of the matrix $A^*$.

We aim to find the perturbation matrix $P = [P_{ij}], \ P_{ij} \in \{0,1 \}$ such that $A^* = A + P$ is a balanced network. 
This indicates that the sums over the columns and the sums over the rows of $A^*$ are equal, i.e., $\sum_j A_{ij} + P_{ij} = \sum_j A_{ji} + P_{ji}, \ \forall i$.
The entries of $P_{ij}$ are restricted such that if a link is present, $P$ cannot add it, i.e., we consider the simple graphs. 
This means $P_{ij} = 0$ if $A_{ij} \neq 0$.
The following linear integer optimization in $P$ finds the minimum number of links to be added,
\begin{align} 
\begin{split}
    {\min_{P_{ij}}} \quad & \sum_i \sum_j P_{ij} \\
    \textnormal{subject to} \quad & A^* = A + P \\
    & \sum_j A^*_{ij} = \sum_i A^*_{ij} \\
    & P_{ij} = 0 \ \, \textnormal{if} \ \, A_{ij} \neq 0, \quad P_{ij} \in \{0,1 \}.
\end{split}
\end{align}
By vectorizing $P$, the above optimization can be rewritten in the vectorized format.
As a result, a set of linear equality constraints $K_{eq} \bp = \bk_{eq}$ is obtained, and $\bp = [p_i] = \vec(P)$.
The following linear integer programming in $\bp$ is solved,
\begin{align} \label{eq:linkadd}
\begin{split}
    {\min_{\bp}} \quad & \bone_{n^2}^\top \bp \\
    \textnormal{subject to} \quad & K_{eq}\bp = \bk_{eq} \\
    & p_i \in \{0,1 \}.
\end{split}
\end{align}
\begin{remark}
We call a trivial solution for $P$ the one that makes the adjacency matrix $A$ symmetric by adding the directed link $(v_i, v_j)$, if $(v_j, v_i) \in \mathcal{E}$ and $(v_i, v_j) \notin \mathcal{E}$, $\forall i,j$.
\end{remark}

\subsection{Link removal} \label{sec:rem}
We can also consider an alternative procedure in which a minimum number of links is removed from a given network to make it minimally reactive, similar to the link addition procedure described in Sec.\,\ref{sec:add}.
We aim to find the matrix $Q = [Q_{ij}]$, $Q_{ij} \in \{0, 1 \}$, such that $A^* = A - Q$ is balanced.
The objective is to satisfy $\sum_j A_{ij} - Q_{ij} = \sum_j A_{ji} - Q_{ji}, \ \forall i$.
Also, when a link is absent in $A$, the corresponding entry in $Q$ should be zero, i.e., $Q_{ij} = 0$ if $A_{ij} = 0$.
The following linear integer optimization in $Q$ finds the minimum number of links to be removed,
\begin{align}  \label{eq:Q}
\begin{split}
    {\min_{Q_{ij}}} \quad & \sum_i \sum_j Q_{ij} \\
    \textnormal{subject to} \quad & A^* = A - Q \\
    & \sum_j A^*_{ij} = \sum_i A^*_{ij} \\
    & Q_{ij} = 0 \ \, \textnormal{if} \ \, A_{ij} = 0, \quad Q_{ij} \in \{0,1 \}.
\end{split}
\end{align}
By vectorizing $Q$, the above optimization can be rewritten in the vectorized format.
Similar to the previous section, the constraints in \eqref{eq:Q} can be rewritten in terms of a set of linear equality constraints in $\bq = \vec(Q)$, i.e., $K_{eq}\bq = \bk_{eq}$.
Hence, a linear integer program in $\bq$ is obtained,
\begin{align} \label{eq:linkrem}
\begin{split}
    {\min_{\bq}} \quad & \bone_{n^2}^\top \bq \\
    \textnormal{subject to} \quad & K_{eq}\bq = \bk_{eq} \\
    & q_i \in \{0, 1 \}.
\end{split}
\end{align}
\begin{remark}
We call a trivial solution for $Q$ the one that makes the adjacency matrix $A$ symmetric by removing all the directed links $(v_j, v_i) \in \mathcal{E}$ if $(v_i, v_j) \notin \mathcal{E}$, $\forall i,j$.
\end{remark}

\subsection{Link addition and removal} \label{sec:addrem}
Here, we aim to find $P$ and $Q$ such that $A^* = A + P - Q$ is balanced.
Hence, the condition $\sum_j A_{ij} + P_{ij} - Q_{ij} = \sum_j A_{ji} + P_{ji} - Q_{ji}, \ \forall i$ must be satisfied.
Similarly, $P_{ij} = 0$ if $A_{ij} \neq 0$ and $Q_{ij} = 0$ if $A_{ij} = 0$.
The following linear integer optimization in $P$ and $Q$ finds the minimum number of links to be either added or removed,
\begin{align} 
\begin{split}
    {\min_{P_{ij}, Q_{ij}}} \quad & \sum_i \sum_j P_{ij} + Q_{ij} \\
    \textnormal{subject to} \quad & A^* = A + P - Q \\
    & \sum_j A^*_{ij} = \sum_i A^*_{ij} \\
    & P_{ij} \, = 0 \quad  \textnormal{if} \quad A_{ij} \neq 0, \quad P_{ij} \in \{0,1 \}, \\
    & Q_{ij} = 0 \quad \textnormal{if} \quad A_{ij} = 0, \quad Q_{ij} \in \{0,1 \}. \\
\end{split}
\end{align}
We form the variable vector $[\bp^\top \ \bq^\top]^\top$, where $\bp = \vec(P)$ and $\bq = \vec(Q)$.
The above constraints are written in the vectorized form as $K_{eq} [\bp^\top \ \bq^\top]^\top = \bk_{eq}$.
To find the optimal $P$ and $Q$, we solve
\begin{align} \label{eq:linkaddrem}
\begin{split}
    {\min_{\bp, \bq}} \quad & J = \left[ \bone_{n^2}^\top \ \bone_{n^2}^\top \right] \begin{bmatrix}
        \bp \\ 
        \bq
    \end{bmatrix} \\
    \textnormal{subject to} \quad & K_{eq} \begin{bmatrix}
        \bp \\ 
        \bq
    \end{bmatrix} = \bk_{eq} \\
    & p_i, q_i \in \{0, 1 \}.
\end{split}
\end{align}
 In general, the solution returned by \eqref{eq:linkaddrem} yields an optimal value of the objective function that is lower or at worst equal to those of the solutions returned by either \eqref{eq:linkadd} and \eqref{eq:linkrem}.

 \begin{remark}
      \textcolor{black}{The optimization in \eqref{eq:linkaddrem} may have more than one optimal solution, all with the same objective value.}
     It is then possible to slightly modify the optimization problem to have it \textcolor{black}{select} a solution \textcolor{black}{among all the optimal solutions} for which adding links is more favorable than removing links. That is,
    \begin{align}
    \begin{split}
        {\min_{\bu}} \quad & [(\bone_{n^2} - \epsilon)^\top \ (\bone_{n^2} + \epsilon)^\top ]  \begin{bmatrix}
            \bp \\ \bq
        \end{bmatrix} \\
        \textnormal{subject to} \quad & K_{eq}\begin{bmatrix}
            \bp \\ \bq
        \end{bmatrix} = \bk_{eq} \\
        & p_i, q_i \in \{0, 1 \},
    \end{split}
    \end{align}
    where $0 < \epsilon < \frac{1}{n^2}$ is a tunable parameter..
 \end{remark}
 \color{black}
\begin{remark}
    Integer linear programs  are known to be NP-Hard, which means that an algorithm that runs in polynomial time in the worst case is unlikely to exist. 
    This does not imply that integer linear programs are unsolvable but that solution times are unlikely to scale well with even the most sophisticated of general solvers, see \cite{gurobi,miplib}. 
    We note this does not necessarily imply that problems \eqref{eq:linkadd}, \eqref{eq:linkrem}, and \eqref{eq:linkaddrem} are NP-Hard as they are specific versions of integer linear programs and polynomial-time algorithms may exist that can leverage their specific structure.
    
    Convex quadratic programming is polynomial-time solvable which is the form of the problem \eqref{eq:optp}. 
    Determining whether the problems have unique optimal solutions is difficult for problems \eqref{eq:linkadd}, \eqref{eq:linkrem}, and \eqref{eq:linkaddrem} and is only guaranteed for convex quadratic programs when the problem is strictly convex and possibly for problems when they have some structure.
\end{remark}

\begin{remark}
    If the network topology and the weights associated with the network connections cannot be changed, the reactivity still can be modified through the coupling strength $\sigma > 0$:
    \begin{equation}
        \dot{x}_i (t) = \sigma \sum_j L_{ij} x_j.
    \end{equation}
    The effective reactivity and the entire spectrum of the Laplacian matrix will then be rescaled by $\sigma$. 
    This may lead to a longer convergence time for reaching the consensus state.
\end{remark}

\color{black}
\subsection{Example}
In this section, we provide an example of graph optimization by link addition or removal, or both, in order to make such graphs minimally reactive.
Figure\,\ref{fig:link} (a) shows a randomly generated unweighted directed graph. 
We use the proposed link perturbation procedures to make the network minimally reactive.
Panels (b), (c), and (d) in Fig.\,\ref{fig:link} show the resulting optimized topologies by link addition \eqref{eq:linkadd}, link removal \eqref{eq:linkrem}, and link addition and removal \eqref{eq:linkaddrem}, respectively.
\begin{figure}
    \centering
    \includegraphics[width=0.55\linewidth]{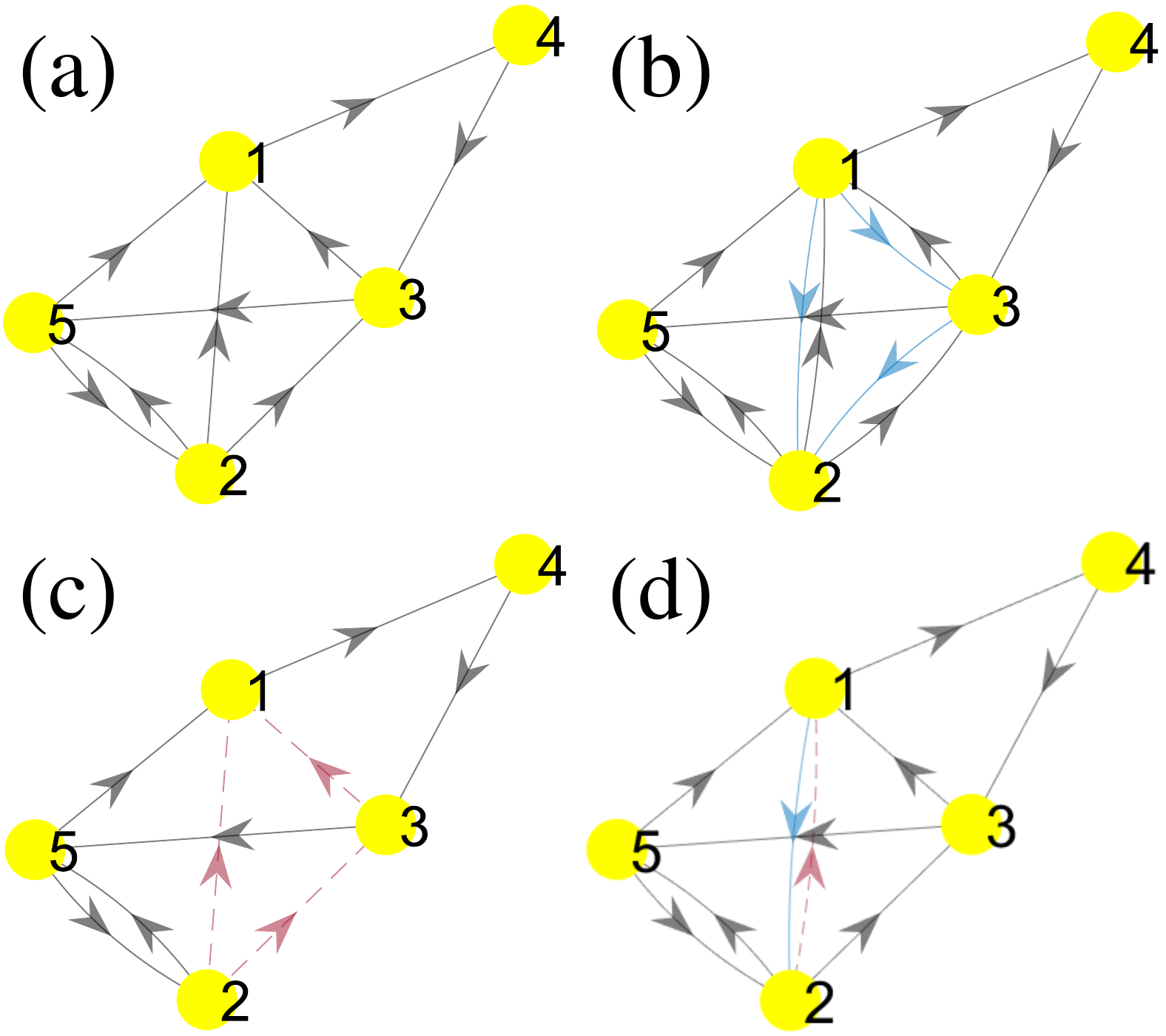}
    \caption{Optimizing (a) the unweighted and directed graph with (b) link addition, (c) link removal, and (d) link addition and removal. 
    The added links are shown in black color and the removed ones are shown in red (dashed lines.)}
    \label{fig:link}
\end{figure}
We see in Fig.\,\ref{fig:link} (b) and (c) that optimal topologies are generated by adding a minimum of three links and removing a minimum of three links, respectively.
The trivial solutions for link addition and link removal require seven link additions (by making all links bi-directional) and seven link removals (by removing all uni-directional links, respectively.)
On the hand, the link addition and removal optimization in panel (d) provides a solution in which a minimally reactive network is generated by performing only two link perturbations (i.e., one link addition and one link removal.)

\subsection{Structural reactivity}
The optimization procedures proposed in Sec.\,\ref{sec:addrem} find the minimum number of link perturbations (additions and/or removals) that make a network minimally reactive.
Here, we \textcolor{black}{introduce a practical measure `structural reactivity'} which assesses how much a given network is far from becoming minimally reactive by link perturbations.
\begin{definition}
    The structural reactivity of a given graph $G$ is the ratio between the minimum number of link perturbations and the number of nodes of the graph.
    Denoting the optimal objective function value from \eqref{eq:linkaddrem} by $J^*$, then the structural reactivity is
    \begin{equation} \label{eq:psi}
        \psi = \frac{J^*}{n}.
    \end{equation}
\end{definition}
We study $\psi$ for 
\textcolor{black}{directed} Erdős–Rényi  graphs (random graphs).
Figure\,\ref{fig:psimean} shows $\psi$ versus the connection probability $p$ of the random graph. 
\begin{figure}
    \centering
    \includegraphics[width=0.6\linewidth]{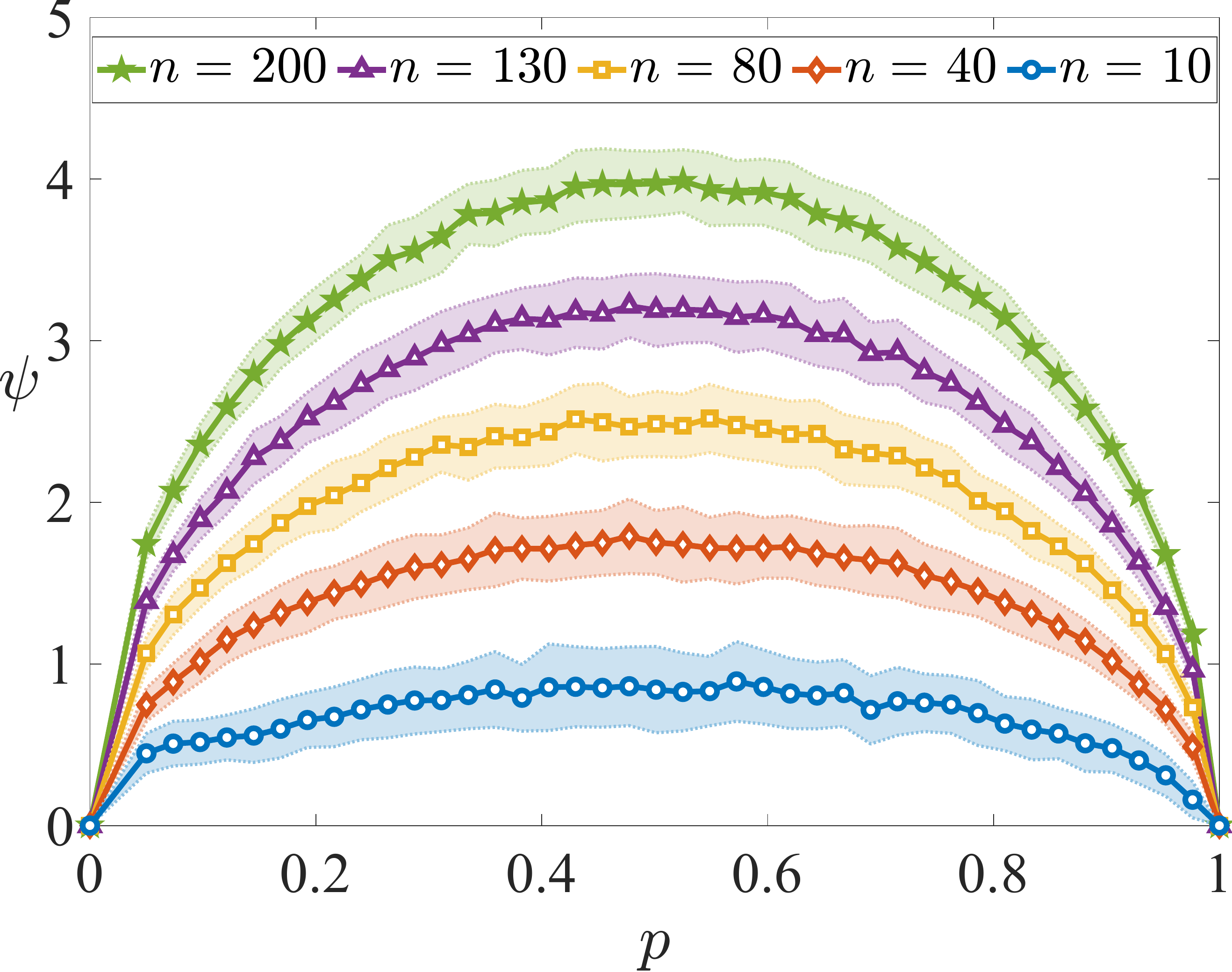}
    \caption{The mean of the structural reactivity $\psi$ is plotted against the connection probability $p$ of Erdős–Rényi graph (random graph) with $n$ nodes.
    For each pair of $p$ and $n$, 100 random graphs are generated and for each graph, $\psi$ is calculated using \eqref{eq:psi}, and here the mean value of $\psi$ over the 100 graphs is plotted. The shaded colored area corresponds to the standard deviation of the data points.}
    \label{fig:psimean}
\end{figure}
When $p = 0$, the corresponding graph is a null graph, which is minimally reactive since $\deg_i^{in} = \deg_i^{out} = 0, \ \forall v_i$. 
Therefore, no links are added or removed.
We see that as $p$ increases, the structural reactivity $\psi$ increases until around $p \approx 0.5$ when $\psi$ reaches its highest value.
For larger values of $p$, $\psi$ decreases until eventually it reaches 0 again when $p=1$ (i.e., fully connected graph.)

We call $G$ a graph and $H$ its complement. We note that $H$ is balanced if and only if $G$ is balanced.
It requires $J^* = a + r$ link perturbations to make $G$ balanced, where $a$ ($r$) is the number of optimally added (removed) links. 
While the graph $H$ requires $r$ ($a$) additional (removal) number of links to become balanced (i.e., for each added  (removed) link in $G$, the corresponding link in $H$ should be removed (added).)
Since the value of $\psi$ for a graph and its complement is the same, the curve of $\psi$ vs. $p$ should be symmetric about $p=0.5$ and, we expect the maximum to be achieved at $p=0.5$.
This is consistent with what is seen in Fig.\,\ref{fig:psimean}, i.e., the shape of the curve $\psi$ vs. $p$ for directed unweighted random graphs is concave, reaches a maximum around $p=0.5$, and is symmetric about $p=0.5$.

\section{Conclusions} \label{sec:conclusion}
In this letter, we have studied the consensus dynamics of directed networks and shown that a class of networks for which the indegrees and outdegrees of all the nodes are the same is minimally reactive. 
Several optimization procedures have been proposed to perturb a given graph by either adjusting the weights of the links or by adding/removing links. 

A new measure of structural reactivity has been introduced which determines how much a given graph is far from becoming minimally reactive by link perturbations. 
A study of random graphs has shown that their structural reactivity is largest when the connection probability is around 0.5.

\appendices
\section{Proof of Theorem \ref{theo2}} \label{appendix}
\textcolor{black}{\textit{(Necessity):}}
Let $L$ be given with $G=(\cV,\cE)$ as its associated graph. Recall from Assumption \ref{assump:G} that an out-directed spanning tree rooted at a node $r \in \cV$ exists. Assume there is a solution $P$ that satisfies  \eqref{eq:optp}. We first show that the solution can be interpreted as a set of $|\cE|$ directed cycles on $G$. For each node $i \in \cV$, the feasibility conditions \eqref{eq:rowsum} and \eqref{eq:columnsum} imply that
\begin{equation}
\label{eq:flowbalance}
\sum_{(i,j) \in \cE} (L_{ij}+P_{ij}) - 
\sum_{(j,i) \in \cE} (L_{ji} + P_{ji}) = 0.
\end{equation}
By interpreting the weight of each directed edge as a flow, we see these two conditions amount to ensuring that the flow into a node is equal to the flow out of a node. 
Equation \eqref{eq:flowbalance} are also called {\it flow-balance} constraints~\cite{amo}. 
Note that setting $P_{ij} = -L_{ij}$ for $(i,j) \in \cE$ results in a zero flow that satisfies \eqref{eq:flowbalance} (and also \eqref{eq:rowsum} and \eqref{eq:columnsum}). 
However, note that such a solution is not feasible to the condition that $L_{ij} + P_{ij} \geq \epsilon$. 
Then by the Augmenting Cycle Theorem \cite[Theorem 3.7]{amo}, the problem of finding a feasible solution to \eqref{eq:optp} is equivalent to finding a flow with at least $\epsilon$ flow along each arc that can be decomposed into at most $|\cE|$ directed cycles with positive flow. 
For the remainder of the proof, let $x_{ij} = L_{ij} + P_{ij}$ for $(i,j) \in \cE$ denote the flow associated with the feasible solution. 
Cycle decomposition means that {\it every arc is on a cycle}. 

Consider the root node $r$ and let any node $j \in \cV$ be given. We claim there is also a directed path from $j$ to $r$. Note that this will then show the theorem as then, for any nodes $i$ and $j$, there is a path from $i$ to $r$ and another path from $r$ to $j$, i.e., the graph is strongly connected.
Let $\{r,i_1,\ldots,i_k=j\}$ denote the directed path from $r$ to $j$. 
We prove there is a $j$-$r$ path by induction on $k$, the length of the path. As a base case, i.e., $i_1=i_k=j$ note that $(r,i_1=j)$ is on one of the cycles formed by $x$ which then forms the path from $i_1$ to $r$. 
For the inductive case, suppose that there is a directed path back from all nodes on paths of length $k \geq 1$ from $r$. 
If $j$ is on a path of length $k+1$, consider the edge $(i_k,i_{k+1}=j)$. This edge is also on a cycle so there is a path from $i_{k+1}$ to $i_k$. Then, by the inductive assumption, there is a path from $i_k$ to $r$ which results in a path from $j$ to $i_k$ to $r$. 
Thus, by induction, there is a path from $j$ to $r$ for every node $j$ in $G$.

\textcolor{black}{\textit{(Sufficiency):}
Let $L$ be given along with its associated graph $G = (\cV,\cE)$. We construct a solution $P$ that is feasible (but almost certainly not optimal) to  \eqref{eq:optp} using the following algorithm. Initialize all perturbations, $P_{ij}$ to zero. We iterate over each edge $(i,j) \in \cE$ as follows. As $G$ is strongly connected there is a directed path from $j$ to $i$ which we denote by $\{j=u_0,u_1,\ldots,u_k=i\}$. For $t=0$ to $k-1$, increment $P_{u_t,u_{t+1}}$ by $L_{ij}$. At the conclusion of the algorithm, we have a new feasible flow composed of $|\cE|$ directed cycles: one for each edge $(i,j) \in \cE$, each with $L_{ij}$ flow.}
 


\end{document}